\documentclass{article}
\usepackage[english]{babel}
\usepackage[letterpaper,top=2cm,bottom=2cm,left=3cm,right=3cm,marginparwidth=1.75cm]{geometry}
\usepackage{amsmath}
\usepackage{graphicx}
\usepackage{subcaption}

\usepackage[colorlinks=true, allcolors=blue]{hyperref}
\usepackage{algorithm}
\usepackage{algorithmic}
\usepackage{mathrsfs}

\usepackage{amsmath,amsfonts,amsthm}
\usepackage{bm}
\usepackage[usenames,dvipsnames]{color}
\usepackage{comment}

\usepackage{amsmath,amsfonts}
\usepackage{algorithmic}
\usepackage{enumitem}
\usepackage{float}
\usepackage{graphicx}
\usepackage{textcomp}
\usepackage{xcolor}
\usepackage{tikz}
\usepackage{listings}
\usepackage{url}
\usepackage{cancel}
\usepackage{xcolor}
\usepackage{soul}
\usepackage{subcaption}
\usepackage[super]{nth}
\usepackage {tikz}
\usepackage[T1]{fontenc}
\usepackage[utf8]{inputenc}
\usepackage{authblk}

\usetikzlibrary {positioning}
\definecolor{honestColor}{cmyk}{0.96,0,0,0}
\definecolor{adversarialColor}{rgb}{1.0, 0.13, 0.32}

\AtBeginDocument{%
  \providecommand\BibTeX{{%
    \normalfont B\kern-0.5em{\scshape i\kern-0.25em b}\kern-0.8em\TeX}
  }
}

\newtheorem{theorem}{Theorem}
\newtheorem{remark}{Remark}

\newtheorem{prop}[theorem]{Proposition}

\newtheorem{example}{\textbf{Example}}

\title{Condorcet Attack Against Fair Transaction Ordering}
\date{}

\begin{document}

\author[1]{Mohammad Amin Vafadar\thanks{vafadar@ualberta.ca}}
\author[1]{Majid Khabbazian\thanks{mkhabbazian@ualberta.ca}}
\affil[1]{Department of Electrical and Computer Engineering, University of Alberta}

\maketitle

\begin{abstract}
    We introduce the \emph{Condorcet attack}, a new threat to fair transaction ordering.
    Specifically, the attack undermines batch-order-fairness, the strongest notion of transaction fair ordering proposed to date.
    The batch-order-fairness guarantees that a transaction $\texttt{tx}$ is ordered before $\texttt{tx}'$ if a majority of nodes in the system receive  $\texttt{tx}$ before $\texttt{tx}'$;
    the only exception (due to an impossibility result) is when $\texttt{tx}$ and $\texttt{tx}'$ fall into a so-called ``Condorcet cycle''.
    When this happens, $\texttt{tx}$ and $\texttt{tx}'$ along with other transactions within the cycle are placed in a batch, and any unfairness inside a batch is ignored. 
    
    In the Condorcet attack, an adversary attempts to undermine the system's fairness by imposing Condorcet cycles to the system.
    In this work, we show that the adversary can indeed impose a Condorcet cycle by submitting as few as two otherwise legitimate transactions to the system.    
    Remarkably, the adversary (e.g., a malicious client) can achieve this even when all the nodes in the system behave honestly.      
    A notable feature of the attack is that it is capable of ``trapping'' transactions that do not naturally fall inside a cycle, i.e. those that are transmitted at significantly different times (with respect to the network latency).
    To mitigate the attack, we propose three methods based on three different complementary approaches.  
    We show the effectiveness of the proposed mitigation methods through simulations, and explain their limitations.
\end{abstract}

\section{Introduction}

    The first blockchain application, Bitcoin, emerged in the midst of the financial crisis of 2008, caused in part by the excessive trust placed in centralized institutions.
    Blockchain technology changed this.
    In blockchain, there is no central authority or intermediary controlling the entire system.
    Instead, transactions are validated and included through a consensus mechanism among the participating parties.
    Decentralization also promotes transparency and reduces the possibility of fraud or corruption since all transactions are publicly recorded and visible to all participants on the network.    

    Despite the decentralized nature of blockchain systems, the ordering of transactions is carried out in a centralized manner; the miner/validator who creates a block determines the ordering of transactions within the block.
    This gives too much power to a single entity as the success and profitability of a transaction can be determined by the order in which the transaction appears in a block~\cite{FlashBoys, BaumCDFG21, EskandariMC19, abs-2203-11520, DarkForest}.
    For instance, when a Non-Fungible Token (NFT) is dropped in a given block, transactions positioned earlier in the block have a higher chance of acquiring the NFT compared to those placed later.
    
    To address this issue, several existing works~\cite{Aequitas, Pompe, Themis, QuickFairness, asiapkc/KelkarDK22, Wendy} proposed decentralized methods for handling transaction ordering, where instead of a single node, a committee of nodes collectively decide on the ordering of received transactions. 
    At the core of these methods, each node in the system reports a list of transactions in the order the node has received them.
    The system then generates and agrees on a ``fair'' ordering by taking the reported orderings into account. 

    Finding a fair ordering is not trivial.
    For instance, suppose that for any two transactions $\texttt{tx}_1$ and $\texttt{tx}_2$, we require $\texttt{tx}_1$ to be placed before $\texttt{tx}_2$ if a large majority of nodes in the system claim to have received $\texttt{tx}_1$ before $\texttt{tx}_2$.
    Despite being a primitive requirement, no method can provide a guarantee due to an impossibility result rooted in social choice theory \cite{SocialChoice}.
    As an example, consider a system consisting of three nodes, where each node has received three transactions:
    $\texttt{tx}_1$, $\texttt{tx}_2$, and $\texttt{tx}_3$.
    Suppose the nodes report the ordering as 
    $[\texttt{tx}_1, \texttt{tx}_2, \texttt{tx}_3]$, 
    $[\texttt{tx}_2, \texttt{tx}_3, \texttt{tx}_1]$,
    and 
    $[\texttt{tx}_3, \texttt{tx}_1, \texttt{tx}_2]$.
    In this case, $\texttt{tx}_1$ is reported to be before $\texttt{tx}_2$ by two nodes (i.e. the majority), $\texttt{tx}_2$ is reported to be before $\texttt{tx}_3$ by two nodes, and $\texttt{tx}_3$ is reported to be before $\texttt{tx}_1$ by two nodes. 
    This essentially creates a cycle, referred to as \emph{Condorcet cycle}~\cite{Condorcet}, which prevents any final ordering from respecting the views of the majority on how transactions should be ordered.

    The existing fair ordering methods adopt a relaxed approach to ordering transactions inside a Condorcet cycle.
    For instance,  Cachin et al. in Quick-Fairness~\cite{QuickFairness} do not mention any ordering mechanism for such transactions, and Kelkar et al. in Aequitas~\cite{Aequitas} suggest a simple alphabetical ordering. 
    This relaxed approach is, perhaps, due to two reasons: 1) it is not possible to guarantee fair ordering of transactions inside a cycle; 2) Condorcet cycles occur infrequently in practice, and when they do occur, they usually involve transactions that are received around the same time by the nodes in the system.
    Nevertheless, in this work, we show that Condorcet cycles deserve more attention as they can be created ``artificially'' by adversaries through what we refer to as the \emph{Condorcet attack}. 
    An interesting feature of the Condorcet attack proposed in this work is that it can be conducted by a client outside the system. In particular, the attack can be effectively executed even when all the nodes in the system are honest!

    As will be explained later, in the Condorcet attack, an adversarial client sends a small number of transactions to different nodes in the system.
    The adversary chooses the timing and order of these transactions to create a Condorcet cycle that traps many honest transactions in it.
    (a Condorcet cycle with only the adversary's transactions in it is all but harmless to the system.)
    This cycle has to be broken, by the leader in a leader-based method, in order to establish a total ordering.
    Even if the leader is honest, the act of breaking the cycle could change the order of honest transactions, which would have otherwise been appropriately ordered\footnote{Kelkar et al.~\cite{Themis} consider it a success for an adversary if the adversary places two transactions into the same cycle when they should not have been.}. 
    
    Defending the Condorcet attack is not straightforward.
    It is partly because it is challenging to differentiate between honest transactions and otherwise valid transactions that are submitted with the intention of creating a cycle. 
    It becomes notably more challenging to safeguard the system when, in addition to the adversarial client outside the system, the leader and possibly a fraction of the nodes in the system are adversarial.
    Nevertheless, in this work, we propose three mitigation techniques based on three different approaches. 
    The proposed techniques complement each other and can work together harmoniously to maximize resistance against the attack. 

    In summary, we make the following contributions.
    We introduce a framework for a new type of attack (Condorcet attack) against fair transaction ordering.
    We show that the attack can be highly successful in trapping honest transactions in a cycle.
    To mitigate the attack, we propose three techniques based on three different complimentary approaches, and show the effectiveness of the technique through simulations.

\section{Related Work}

    The classical approach to mandating fair transaction ordering is through \emph{secure causal ordering}, a method introduced by Birman and Reiter in 1994~\cite{ReiterBirman}, and later improved by Cachin et al. in 2001~\cite{CachinKPS01}.
    This method uses encryption to conceal the content of transactions during the ordering process, and allows decryption of transactions only after the order of transactions is finalized. 
    This prevents an adversary from observing the content of transactions during the ordering process, thereby effectively eliminating attacks such as the sandwich attack~\cite{DarkForest} that rely on inspecting transaction contents.  
    However, the method is unable to prevent ``blind front-running attacks'' where, for instance, the adversary's sole objective is to order her transaction first (to, for example, get priority in purchasing a token). 
    In addition, the method cannot prevent attacks based on transactions' metadata, as metadata (such as the source of a transaction) is not encrypted.

    The second approach to mandating fair transaction ordering involves a first-come, first-served strategy. 
    This approach is complementary to the first approach and has been the focus of several recent studies.
    The existing methods that follow this strategy can be broadly classified into two categories: timestamp-based methods and batch-based methods. 
    Timestamp-based methods are computationally inexpensive but require synchronized clocks.
    Batch-based methods, on the other hand, offer stronger fairness than timestamp-based methods, but can tolerate fewer adversarial nodes.    
    
    \textbf{Timestamp-based Methods.}
    An example of a timestamp-based protocol is Pompe~\cite{Pompe} due to Zhang et al. Pompe introduces a notion of fairness called the \emph{ordering linearizability}.
    This notion stipulates that if the highest timestamp of a transaction \texttt{tx} is less than the lowest timestamp of a transaction $\texttt{tx}'$ among honest nodes, then \texttt{tx} must be ordered before $\texttt{tx}'$ in the final order of transactions.  
    Although can enforce ordering linearizability, Pompe suffers from censorship issues, as noted in~\cite{Themis}.
    
    Kursawe's Wendy protocol~\cite{Wendy} is another timestamp-based protocol that defines a notion of fairness called \emph{timed-relative-fairness}.
    This notion requires that if all honest nodes received a transaction \texttt{tx} before time $\tau$, and transaction $\texttt{tx}'$ after $\tau$, then \texttt{tx} must be ordered before $\texttt{tx}'$.

    \textbf{Batch-based Methods.}
    Aequitas~\cite{Aequitas} by Kelkar et al. is a batch-based method proposed for fair transaction ordering.
    Aequitas enforces a fairness notion known as the \emph{$\gamma$-batch-order-fairness}.
    The notion requires that if two transactions \texttt{tx} and $\texttt{tx}'$ are received by all nodes in a system with $n$ nodes, and $\gamma n$ nodes received \texttt{tx} before $\texttt{tx}'$, then all honest nodes must output \texttt{tx} no later than $\texttt{tx}'$. 
    Aequitas suffers from high communication complexity of $\mathcal{O}(n^3)$, and can guarantee only a weak notion of liveness, one of the two pillars of consensus security.
    
    The second batch-based method is Quick-Fairness~\cite{QuickFairness} proposed by Cachin et al. 
    This method enforces a fairness notion called the \emph{$\kappa$-differential order-fairness}.
    This notion mandates that if the number of nodes that have received transaction \texttt{tx} before $\texttt{tx}'$ exceeds $\kappa + 2f$ for some $\kappa \geq 0$, then \texttt{tx} should be ordered no later than $\texttt{tx}'$, where $f$ is the maximum number of adversarial nodes in the system. 
    Kelkar et al.~\cite{Themis} show that this notion of fairness is indeed a re-parameterized version of the $\gamma$-batch-order-fairness notion.
    They also demonstrate that the Quick-Fairness protocol satisfies fairness only when all nodes are honest.

    Kelkar et al. addressed the shortcomings of Aequitas in their protocol called Themis~\cite{Themis}.
    Themis satisfies the $\gamma$-batch-order-fairness notion, and solves the liveness problem of Aequitas. 
    Moreover, it offers a communication complexity of $\mathcal{O}(n^2)$ instead of $\mathcal{O}(n^3)$ offered by Aequitas.
    In addition, it satisfies a more generalized notion of fairness than the one used in Quick-Fairness and a stronger notion of fairness than those used in the existing time-based methods. For these reasons, in our work, we focus on Themis as the state-of-the-art fair transaction ordering method.

\section{Model}
  
\textbf{System.}
    We consider a permissioned system with a committee of $n$ nodes.
    The nodes receive transactions directly from clients, and submit the list of their received transactions together with the order in which the transactions were received to a special node called the leader. The leader collects the lists of transactions from the nodes, and proposes a final ordering using a pre-decided fair-ordering protocol.
    The leader in the system is not fixed, and can change through a pre-determined protocol.

\medskip
\noindent
\textbf{Fair Ordering.}
    We adopt the batch-order-fairness from~\cite{Aequitas, Themis}, the strongest notion of fair ordering proposed to date.
    For a parameter $\frac{1}{2} < \gamma \leq 1$, the batch-order-fairness specifies that if a fraction $\gamma$ of nodes receive a transaction $\texttt{tx}$ before receiving another transaction $\texttt{tx}'$, then $\texttt{tx}$ must be placed in the order before $\texttt{tx}'$, with exceptions allowed only if $\texttt{tx}$ and $\texttt{tx}'$ are within a same Condorcet cycle (Condorcet cycles are defined in Section~\ref{sec:pre}). 
    Transactions within a cycle are placed in a batch, and are ordered by a method that we refer to as \emph{batch-ordering scheme}. 
    The existing fair ordering protocols either do not specify a batch-ordering scheme or propose a simple one (e.g., an alphabetical-based scheme~\cite{Aequitas}).

\medskip
\noindent
\textbf{Network.}
    The network utilizes public key infrastructure and secure digital signatures for communications. 
    As in~\cite{Aequitas}, we consider two networks: the (standard) internal network (for communication amongst nodes in the system) and the external network (for clients to transmit their transactions to the system). 

    We assume that the network operates under partial synchrony \cite{PartialSynchrony}, meaning that there is a network delay $\Delta$ (not known to the nodes) that limits the amount of time it takes for messages to be delivered between nodes.

\medskip
\noindent
\textbf{Adversary.}    
    We consider an adversary who has control over $f \geq 0$ out of $n$ nodes, and also possesses at least one client capable of submitting transactions to the system.
    The adversary can deviate arbitrarily from the protocol.
    The adversary does not have control over the external network, but may have full control over the internal network, hence can delay and reorder messages up to the bound $\Delta$.

\section{Preliminaries}   
\label{sec:pre}
    
\noindent
\textbf{Graph Terminology.}
    We use $G = (V, E)$ to denote a graph with the set of vertices $V$ and the set of edges $E$. 
    In this work, each vertex represents a transaction, therefore, we use the terms vertices and transactions interchangeably. 
    Unless otherwise specified, we use an unweighted and directed graph.
    In the case of a weighted graph, the weight or cost associated with the edge $(u, v) \in E$ is represented by $w(u, v)$.

    A \textit{tournament graph} is a directed graph where every pair of distinct vertices is connected by a directed edge in either of two possible directions.
    A \textit{Strongly Connected Component (SCC)} in a graph is a maximal subgraph in which there is a path from every vertex to every other vertex.
    A \textit{condensation graph} is obtained from the original graph by combining its SCCs into a single vertex.
    A \textit{Directed Acyclic Graph (DAG)} is a directed graph that contains no cycles, meaning it is possible to move from one vertex to another along the directed edges, but it is not possible to return to the original vertex by following a sequence of directed edges.
    A \textit{topological sort} is an ordering of the vertices in a DAG such that for every directed edge $(u, v)$, vertex $u$ appears before vertex $v$.
    In other words, if there is a directed edge from vertex $u$ to vertex $v$, then $u$ must appear before $v$ in the topological sort.
    A \textit{Hamiltonian Path} is a path in a graph that passes through every vertex exactly once.
    A \textit{Hamiltonian Cycle} is a cycle in a graph that passes through every vertex exactly once.

\medskip
\noindent
\textbf{Themis.}    
    Themis is the latest ordering method which achieves batch-order-fairness in the presence of an adversary who controls up to $f < \frac{(2\gamma - 1)n}{4}$ nodes out of $n$ nodes.
    Themis is a leader-based method and works in three phases, as described below.
    
\begin{itemize}
    \item Phase 1 (Fair Propose):
    The Fair Propose phase is the first phase of the algorithm, where each node proposes a set of transactions and their local orderings to the leader. 
    The leader then uses the local orderings of $n - f$ nodes to build a dependency graph.
    In the dependency graph, an edge from a vertex $v_1$ to $v_2$ indicates that the transaction $v_1$ should be placed before the transaction $v_2$. 
    From the dependency graph, the leader then computes the condensation graph and its topological sorting to output a fair ordering.
    
    \item Phase 2 (Fair Update):
    The Fair Update phase is the second phase of the algorithm, where nodes update the ordering for previous proposals. 
    This is necessary as new transactions may depend on previously proposed transactions, and these dependencies need to be accounted for in the ordering.
    The Fair Update algorithm takes the local transaction orderings of $n - f$ nodes for previously proposed shaded transactions as input and outputs the updated dependencies.
    
    \item Phase 3 (Fair Finalize):
    The Fair Finalize phase is the third and final phase of the algorithm, where a sequence of proposals is finalized into a final ordering.
    The Fair Finalize algorithm updates the graphs for each proposal and computes the condensation graphs and their topological sorting.
    It then retrieves the final transaction ordering for each proposal based on the Hamiltonian cycles of the vertices in the sorted condensation graphs.
\end{itemize}

\medskip
\noindent
\textbf{Condorcet Cycles.}  
    As mentioned above, Themis constructs a dependency graph, a directed graph where each vertex represents a transaction, and an edge from a vertex $v_1$ to $v_2$ indicates that the transaction corresponding to $v_1$ should be placed before the transaction corresponding to $v_2$. 
    We refer to any cycle in this dependency graph as a \emph{Condorcet cycle}.
    We note that cycles can occur in a dependency graph because of the Condorcet paradox~\cite{Themis}.

\section{Condorcet Attack}

    In this section, we present the framework of the Condorcet attack.
    The attack aims at trapping honest transactions (i.e., transactions submitted by honest clients) inside a Condorcet cycle. 
    If there is no effective batch-ordering scheme in place (e.g., if the batch-ordering scheme is alphabetical-based as suggested in~\cite{Aequitas}), this can change the ordering of the honest transactions even when all the nodes in the system are honest.
    
    An adversary can take different strategies to impose a Condorcet cycle. 
    For instance, suppose that the adversary controls $f$ nodes, including the leader, in the system. 
    The adversary then controls $f$ local orderings, and can manipulate these orderings in a way to create a cycle. 
    In the simulation section, we show that this strategy can not only create a cycle but also chain the cycles to involve more honest transactions. 
    Nevertheless, the length of these cycles is typically small and the chain usually breaks rather quickly. 
    As a result, this strategy is not effective in trapping distant transactions\footnote{The analysis of why this occurs is left for future work.} (e.g., two transactions whose times of submission are separated by a multiple of the average network latency). 
    
    Another strategy, which is the one we take in this work, is to create a Condorcet cycle by injecting (valid) transactions into the system following a pre-described pattern. 
    This can be done by an adversarial client outside the system, and can be effective even when all the nodes in the system are honest.
    The attack will be more effective in creating cycles and bypassing potential countermeasures if the adversary controls a fraction of nodes in the system.
    
    The immediate damage of imposing a Condorcet cycle, as mentioned earlier, is that it can change the true ordering of honest transactions. 
    In addition to this, the attack may be used to conduct other malicious activities; for instance, the adversary can create a cycle and then with the help of an adversarial leader can try to place its own transaction in desired positions in the final ordering. 

\medskip
\noindent
\textbf{Attack Framework.}
    Let $\mathscr{C}$ be a client controlled by the adversary, and $\mathcal{S}$ be a set of arbitrary but valid transactions created by $\mathscr{C}$.
    Let $\mathcal{P}$ be a partition of the nodes in the system. 
    In its general form, the Condorcet attack is executed in three phases:

\begin{itemize}
    \item Phase 1 (Initialization):
    In this phase, the client $\mathscr{C}$ sends a number of transactions from the set $\mathcal{S}$ to each node in the system.
    The set of transactions sent to a node can be different from that sent to another node.
    More specifically, the client $\mathscr{C}$ assigns a subset $\mathcal{S}_i$ of $\mathcal{S}$ (possibly an empty subset) to each part $\mathcal{P}_i$ in the partition $\mathcal{P}$.
    It then determines an ordering for each subset $\mathcal{S}_i$, and sends the transactions in $\mathcal{S}_i$ to all the nodes in part $\mathcal{P}_i$ with the determined order.

    \item Phase 2 (Pause):
    In the second phase, the attacker waits for a specific amount of time, referred to as the pause time, for the honest transactions to be received by the nodes. 
    The adversary can trap more transactions within a cycle as the pause time increases. 
    However, the pause time should be limited to a single consensus round in the system as the attack should not extend across multiple consensus rounds.
    
    \item Phase 3 (Finalization):
    The third and final phase is the finalization phase, where the attacker completes the Condorcet cycle by sending a new set of transactions to each part in the partition.
    More specifically, the client $\mathscr{C}$ assigns a subset $\mathcal{S}'_i$ of $\mathcal{S}$ (typically a different subset than $\mathcal{S}_i$, used in the initialization phase) to each part $\mathcal{P}_i$ in the partition $\mathcal{P}$.
    It then determines an ordering for each subset $\mathcal{S}'_i$, and sends the transactions in $\mathcal{S}'_i$ to all the nodes in part $\mathcal{P}_i$ with the determined order.    
\end{itemize}

\begin{example}
\label{exm:Cond3nodes}
  Let $\mathcal{P} = \{P_1, P_2, P_3\}$ be a partition of nodes, where $P_1$, $P_2$ and $P_3$ are three parts with almost equal size.
  In this simple example, the adversary $\mathscr{C}$ uses/injects two transactions $\texttt{A}, \texttt{B}$ (i.e., $\mathcal{S} = \{\texttt{A}, \texttt{B}\}$).
  In the initialization phase, $\mathscr{C}$ sends the transaction $\texttt{A}$ and then $\texttt{B}$ to all the nodes in part $P_1$, and sends the transaction $\texttt{B}$ to all the nodes in part $P_2$ (it sends no transactions to the nodes in part $P_3)$.
  Then, after the pause period, $\mathscr{C}$ sends $\texttt{A}$ to all the nodes in part $P_2$, and transaction $\texttt{A}$ and $\texttt{B}$, in that order, to all the nodes in part $P_3$.
  Suppose that during the pause phase, the nodes receive three honest transactions $\texttt{tx}_1$, $\texttt{tx}_2$, and $\texttt{tx}_3$ all the in that order.
  The local ordering of transactions at each node will be then:
  \[
    \begin{array}{cc}
      P_1: & [\texttt{A}, \texttt{B}, \texttt{tx}_1, \texttt{tx}_2, \texttt{tx}_3]\\
      P_2: & [\texttt{B}, \texttt{tx}_1, \texttt{tx}_2, \texttt{tx}_3, \texttt{A}]\\
      P_3: & [\texttt{tx}_1, \texttt{tx}_2, \texttt{tx}_3, \texttt{A}, \texttt{B}]\\      
    \end{array}
  \]
  Note that without the adversarial client $\mathscr{C}$ disturbing the system (i.e., without transactions $\texttt{A}$ and $\texttt{B}$), the system would have had an easy job of ordering the honest transactions as all the nodes in the system have received the honest transactions in the same order, i.e. $[\texttt{tx}_1$, $\texttt{tx}_2$, $\texttt{tx}_3]$.
  Because of the adversary's transactions $\texttt{A}$, $\texttt{B}$, and $\texttt{C}$, however, we have a cycle now as illustrated in Figure~\ref{fig:Cond3nodesFigure}.
  In this figure, an edge from a transaction $\texttt{tx}$ to a transaction $\texttt{tx}'$ indicates that the majority of the nodes have received $\texttt{tx}$ before $\texttt{tx}'$.

    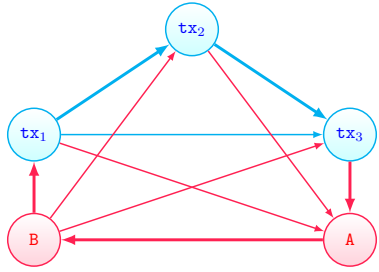
\begin{figure}[h]
        \centering
        \scalebox{0.7}{
            \begin {tikzpicture}[-latex, auto, node distance = 2cm and 3cm, on grid, thick, honest/.style = {circle, top color = white, bottom color = honestColor!20, draw, honestColor, text = blue, minimum width = 1cm}, adversarial/.style = {circle, top color = white, bottom color = adversarialColor!20, draw, adversarialColor, text = red, minimum width = 1cm}]
                \node[honest] (A) {$\texttt{tx}_2$};
                \node[honest] (B) [below left  = of A] {$\texttt{tx}_1$};
                \node[honest] (C) [below right = of A] {$\texttt{tx}_3$};
                \node[adversarial] (D) [below = of B] {\texttt{B}};
                \node[adversarial] (E) [below = of C] {\texttt{A}};
                
                \path (B) [honestColor, ultra thick] edge (A);
                \path (B) [honestColor] edge (C);
                \path (B) [adversarialColor] edge (E);
                \path (D) [adversarialColor, ultra thick] edge (B);
                \path (A) [honestColor, ultra thick] edge (C);
                \path (A) [adversarialColor] edge (E);
                \path (D) [adversarialColor] edge (A);
                \path (C) [adversarialColor, ultra thick] edge (E);
                \path (D) [adversarialColor] edge (C);
                \path (E) [adversarialColor, ultra thick] edge (D);
            \end{tikzpicture}
        }
        \caption{A Condorcet cycle created using two transactions $\texttt{A}$ and $\texttt{B}$}
        \label{fig:Cond3nodesFigure} 
    \end{figure}
\end{example} 

\begin{remark}
    In practice, nodes in the system may receive some honest transactions during the initialization and/or finalization phases. These transactions may or may not get trapped in the Condorcet cycle.
    Based on our simulation results, however, the vast majority of honest transactions during the pause time fall into the Condorcet cycle.
\end{remark}

\begin{remark}
    A potential issue that can impact the success of the Condorcet attack is that the external network may deliver the transactions injected by the adversary out of order.
    For instance, in Example~\ref{exm:Cond3nodes}, the transactions $\texttt{A}$ and $\texttt{B}$ may be received out of order by the nodes in part $P_1$, in which case a cycle does not occur. 
    If the network is prone to packet reordering, then to improve its success, the adversary can execute multiple Condorcet attacks concurrently through what we refer to as \emph{cloning}.     
\end{remark}

\medskip
\noindent
\textbf{Cloning.}
    Packet reordering can happen in a network because of various factors such as network congestion, routing algorithms, and the physical distance between the source and the destination. 
    To conduct a successful Condorcet attack, it is important that nodes receive the injected packets in the order they were transmitted; a deviation from the intended order may result in the failure of the attack.
    
    To increase the success probability of the attack in the presence of network reordering, the adversary can send cloned transactions to the nodes:
    Instead of sending a single transaction $\texttt{A}$, the adversary sends multiple clones of the transaction.
    For instance, in Example~\ref{exm:Cond3nodes}, the adversary can send $\texttt{A}_1$ and $\texttt{A}_2$ instead of $\texttt{A}$, and sends $\texttt{B}_1$ and $\texttt{B}_2$ instead of $\texttt{B}$. Essentially, the adversary interleaves the execution of two Condorcet attacks (for better results, the adversary can interleave several instances of the attack). 
    Then, if the network does not change the order of the transactions, the nodes in parts $P_1$, $P_2$, and $P_3$ will receive transactions as follows:
    \[
        \begin{array}{cc}
            P_1: &  [\texttt{A}_1, \texttt{A}_2, \texttt{B}_1, \texttt{B}_2, \texttt{tx}_1, \texttt{tx}_2, \texttt{tx}_3]\\
            P_2: &  [\texttt{B}_1, \texttt{B}_2,\texttt{tx}_1, \texttt{tx}_2, \texttt{tx}_3, \texttt{A}_1, \texttt{A}_2]\\
            P_3: &  [\texttt{tx}_1, \texttt{tx}_2, \texttt{tx}_3,  \texttt{A}_1, \texttt{A}_2, \texttt{B}_1, \texttt{B}_2]\\      
        \end{array}
    \]
    In Section~\ref{sec:cloning}, we show that cloning can significantly increase the success rate of the Condorcet attack in the presence of network reordering.

\medskip
\noindent
\textbf{Impact on Current Solutions}. 
    The current fair transaction ordering protocols either do not offer a batch-ordering scheme (e.g. \cite{QuickFairness}) or offer a primitive one (e.g. \cite{Aequitas}). 
    For instance, the proposed batch-ordering scheme in Aequitas~\cite{Aequitas} is alphabetical ordering. 
    Therefore, if an adversary creates a Condorcet cycle, as in Example~\ref{exm:Cond3nodes}, the honest transactions will be ordered alphabetically rather than by the time of their arrival.

    Themis~\cite{Themis}, proposes a more thoughtful batch-ordering scheme.
    In this scheme, a Hamiltonian cycle is built and then used to order transactions in the cycle.
    In the best-case scenario, the order of honest transactions in the Hamiltonian cycle is preserved.
    Even in this case, the final ordering of these transactions can change because the Hamiltonian cycle has to be converted into a path by breaking the cycle at one point.
    It is at this point where honest transactions can be divided into two groups.
    The ordering of the honest transactions within each group remains correct, but the ordering of any two transactions from different groups will be incorrect. 
    Therefore, similar to~\cite{QuickFairness} and~\cite{Aequitas}, Themis is vulnerable to the Condorcet attack even if all the nodes (including the leader) in the system are honest.

    To combat the Condorcet attack, a natural approach is to use a strong batch-ordering scheme.
    For instance, in Example~\ref{exm:Cond3nodes}, we can observe that all the nodes report $\texttt{tx}_1$ before $\texttt{tx}_2$, and all the nodes report $\texttt{tx}_2$ before $\texttt{tx}_3$, whereas only two third of the nodes report $\texttt{A}$ before $\texttt{B}$.
    Therefore, in Themis, when it comes to breaking a link in the Hamiltonian cycle, one may choose the weakest link.
    In Example~\ref{exm:Cond3nodes}, the weakest link is between adversarial transactions, and breaking it does not change the true ordering of the honest transactions.
    This solution works for the scenario described in Example~\ref{exm:Cond3nodes}.
    However, as will be explained later, it may not work in other settings, for example when the adversary controls a faction of nodes in the system.

\section{Mitigation}
    Despite its simplicity, it is not straightforward to completely defeat the Condorcet attack.
    In the following, we present three mitigation techniques based on three different approaches to hinder an adversary from successfully executing the attack. 
    We elaborate on the strength of each technique and confirm it through simulations later in Section~\ref{sec:simulation}. 
    We also explain the limitation of each technique, i.e. under what settings/assumptions the technique may not be effective.

    An interesting feature of the proposed mitigation methods is that they do not conflict with each other, thus in practice, they can be applied together for the maximum defense against the attack.
    Another interesting feature of the proposed mitigation methods is that they can be easily applied to Themis, which is currently the strongest fair-ordering solution in the literature.
    We elaborate on this when we cover each proposed mitigation. 

\subsection{Ranked Pairs Batch-ordering}
    \label{sec:ranked}
    The approach we take in our first proposed mitigation is to use a strong batch-ordering scheme to order transactions within a batch.
    Formally, a batch-ordering scheme is a method that takes as input a strongly connected (possibly weighted) directed graph $G = (V, E)$, and returns an ordering of the vertices $V$.
    The strongly connected graph represents the transactions that are in a batch/cycle.    
    
    The candidate for our batch-ordering scheme is \emph{ranked pairs}, an electoral system developed by Nicolaus Tideman in 1987~\cite{RankedPairs}.
    Ranked pairs satisfies many natural and well-studied axiomatic properties in social choice theory\footnote{besides Schulze, ranked pairs is the only existing electoral system that satisfies anonymity, Condorcet criterion, resolvability, Pareto optimality, reversal symmetry, monotonicity, and independence of clones~\cite{Schulze11}.} and is resistance to certain manipulations including adding, deleting and changing a fraction of orderings reported by nodes~\cite{Harvard}. 
    In ranked pairs, the ordering is essentially achieved by choosing a maximal subset $E'$ of $E$ in the inputted graph $G = (V, E)$ with high weights such that $G' = (V, E')$ is a DAG.
    The DAG is then used to establish an ordering of the vertices $V$.
    
    More specifically, our ranked pairs batch-ordering scheme takes as input a weighted directed graph $G = (V, E)$.
    Let $E_1 = E$.
    In step $i$, $i \geq 1$, the algorithm selects an edge $(u, v) \in E_i$ with the highest weight\footnote{When there are multiple edges with the highest weight, one can be chosen according to a fixed tie-breaking method.}.
    It then sets the order $u \prec v$, unless this violates the transitivity of the orders decided in previous steps.
    Finally, it sets $E_{i + 1} \leftarrow E_i \backslash \{(v_i, v_j)\}$, and terminates if $E_{i + 1} = \emptyset$.


    We note that the idea in the above batch-ordering scheme is to establish an ordering using the strongest edges in $G$.
    This will be an effective defense against the Condorcet attack if the ordering of the honest transactions has ``strong support'' in the system.
    In a special case where all the nodes are honest, and all support/report the same ordering of honest transactions, the Condorcet attack can be fully prevented as stated in the following theorem.
    
    \begin{prop}
    \label{prp:paris}
        Suppose that the Condorcet attack succeeds in creating a Condorcet cycle. \\
        Let $\texttt{tx}_1, \texttt{tx}_2, \ldots, \texttt{tx}_m$ be the set of honest transactions in the Condorcet cycle.
        Suppose that all the nodes in the system are honest and report $\texttt{tx}_i$ before $\texttt{tx}_j$ for every $1 \leq i < j \leq m$.
        Then the proposed ranked pairs batch-ordering scheme returns the true ordering of the honest transaction, that is it orders $\texttt{tx}_i$ before $\texttt{tx}_j$ for every $1 \leq i < j \leq m$.
    \end{prop}
    
    \begin{proof}
      Let $G = (V, E)$ be the graph with $V$ representing the transactions in the Condorcet cycle, and the weight of each edge $(u, v) \in E$, represented as $w(u, v)$, be equal to the number of nodes that reported $u$ before $v$.
      Let $u_1, u_2, \ldots, u_m$ be the vertices in $V$ that represent the honest transactions.
      Let $E_f \subseteq E$ be the set of all edges with the full support of the nodes, that is 
      \[
          E_f = \{e \in E | w(e) = n\},
      \]
      where $n$ is the number of nods in the system.
      Since all the nodes in the system have the same view on the ordering of the honest transactions, we get that $(u_i, u_j) \in E_f$ for every $1 \leq i < j \leq m$.
      We note that the sub-graph $G' = (V, E_f)$ of $G$ is cycle free, as otherwise there will be a cycle in the ordering of individual nodes.
      The ranked pairs batch-ordering algorithm first chooses all the edges in $E_f$ before proceeding with other edges in $E$.
      When the algorithm covers all the edges in $E_f$ the true ordering of the honest transactions will be set, and cannot be changed by the remaining steps of the algorithm.      
    \end{proof}

    \medskip
    \noindent
    \textbf{Limitation.}
    Proposition~\ref{prp:paris} considers an ideal scenario where 1) all the nodes are honest, and 2) they all report the honest transaction in the same order.
    If one of the above two conditions does not hold, however, the Condorcet attack may be able to create a cycle (see the following example).

    \begin{example}
        Consider a system with $n = 5$ nodes.
        Let $\texttt{tx}_1, \texttt{tx}_2, \texttt{tx}_3$ be three honest transactions.
        An adversarial client $\mathscr{C}$ can create a Condorcet cycle of the form
        \[
            \begin{array}{cc}
                N_1: &  [\texttt{A}_1, \texttt{A}_2, \texttt{A}_3, \texttt{A}_4,  \texttt{tx}_1, \texttt{tx}_2, \texttt{tx}_3]\\
                N_2: &  [\texttt{A}_2, \texttt{A}_3, \texttt{A}_4,  \texttt{tx}_1, \texttt{tx}_2, \texttt{tx}_3, \texttt{A}_1]\\
                N_3: &  [\texttt{A}_3, \texttt{A}_4,  \texttt{tx}_1, \texttt{tx}_2, \texttt{tx}_3, \texttt{A}_1, \texttt{A}_2]\\  
                N_4: &  [\texttt{A}_4,  \texttt{tx}_1, \texttt{tx}_2, \texttt{tx}_3, \texttt{A}_1, \texttt{A}_2, \texttt{A}_3]\\ 
                N_5: &  [ \texttt{tx}_3, \texttt{tx}_2, \texttt{tx}_1, \texttt{A}_1, \texttt{A}_2, \texttt{A}_3, \texttt{A}_4]\\
            \end{array}
        \]
        where $\texttt{A}_1, \texttt{A}_2, \texttt{A}_3, \texttt{A}_4$ are the transactions submitted by $\mathscr{C}$. Note that all the nodes, except node 5, report the order $[\texttt{tx}_1, \texttt{tx}_2, \texttt{tx}_3]$, while node 5 reports $[\texttt{tx}_3, \texttt{tx}_2, \texttt{tx}_1]$ (Node 5 is either controlled by the adversary or is an honest node who has simply received the transactions in this order).
        If we run the proposed ranked pairs batch-ordering scheme on this cycle, the returned order of honest transactions may be incorrect.
        It is because the edge between any pair of transactions has a weight of 4 in the dependency graph.
        As a result, an edge between two honest transactions such as $\texttt{tx}_1$ and $\texttt{tx}_2$ may be eliminated in the ranked pairs method, which would result in $\texttt{tx}_2$ and $\texttt{tx}_3$ to be ordered before $\texttt{tx}_1$.
    \end{example}

    \begin{remark}
        To use the proposed ranked pairs batch-ordering scheme in Themis, we can simply replace the Hamiltonian-based batch-ordering scheme of Themis with the ranked pairs batch-ordering scheme in the FairFinalize algorithm.
        We remark that the weight information of the dependency graph is available within the FairFinalize algorithm, thus this replacement is possible.
    \end{remark}

\subsection{Post-decryption Resolution}
    In secure causal ordering, as mentioned earlier, transactions are ordered while they are encrypted, and get decrypted only once a total ordering is committed~\cite{ReiterBirman, CachinKPS01}. 
    This prevents an adversary from observing the contents of transactions while they are being ordered, hence eliminating those front-running attacks (e.g. the sandwich attack~\cite{DarkForest}) that must examine the content of transactions.
    
    To mitigate the Condorcet attack, we propose to maintain the above strategy, except we leave the ordering of transactions inside a Condorcet cycle to after they are decrypted.
    Note that after the decryption of these transactions, an adversary cannot impose a change to the ordering as
    1) there is already a consensus on the set of transactions that must be included, thus the adversary cannot add or remove any transaction to the set;
    2) the ordering of the transactions is performed locally at each node using a pre-determined algorithm.
    In other words, it is too late for the adversary to manipulate the ordering of transactions, although the contents of transactions are disclosed.
    
    Once the transactions within a cycle are decrypted, their contents are disclosed, enabling them to be partitioned into independent groups (i.e., transactions inside different groups are independent of each other).
    Each group can then be ordered independent of the others. 
    By implementing this measure, the adversary is unable to manipulate the ordering of honest transactions if the adversary's transactions are independent of honest transactions. 
    This is because the adversary's transactions will not fall within any group that includes honest transactions.
    Note that we still need to order the groups themselves (i.e. which group comes first, which comes second, and so on).
    As transactions across various groups have no effect on one another, the groups can be safely ordered using a pre-determined algorithm such as ranked pairs as described in Section~\ref{sec:ranked}. 
    
    \begin{remark}
        In the Themis protocol, we can apply the above post-decryption resolution method within the FairFinalize algorithm: If transactions $\texttt{A}$ and $\texttt{B}$ are independent, the edge between them in the dependency graph can be safely removed.
    \end{remark}
    
    \textbf{Limitation.} 
        The post-decryption resolution prevents the adversary from manipulating the order of honest transactions if the adversary's transactions are independent of the honest transactions. 
        In certain scenarios, however, the adversary may be able to create dependencies.
        For instance, consider a situation where a popular NFT is dropping in a block currently being formed. 
        Given the high demand, many transactions are transmitted with the intention of acquiring this NFT.
        Recognizing this, the adversary can execute the Condorcet attack by using transactions that fall within the same dependency group as those attempting to acquire the NFT.
        
        Another limitation of the post-decryption resolution is the computational burden it places on the system to identify dependencies between transactions.

\subsection{Broadcast}
    In the Condorcet attack, the adversary follows a well-structured three-phase strategy: in the first phase, the adversary sends a set of transactions, then pauses in the second phase, and then finishes the attack by sending another round of transactions in the third phase.
    The idea behind our third mitigation technique is to disturb/break the above pattern by broadcasting transactions inside the system as soon as they arrive at an honest node.
    Because of the broadcast, the adversary's transactions that were submitted in the first phase will propagate in the system, which can nullify the adversary's target in the third phase since the transactions that the adversary transmits in the third phase have already been received by the nodes (thus their order has already been decided by the nodes).
    
    In Section~\ref{sec:mitigationsim}, we observe that this strategy proves highly effective in mitigating the suggested Condorcet attack. 
    However, it is important to note that this strategy does incur increased communication overhead as a drawback.
    For instance, in Themis, nodes transmit transactions only to the leader as opposed to broadcasting in the network by themselves.
    Therefore, when applied to Themis, the above strategy will increase Themis's communication overhead (although it does not increase Themis's quadratic communication complexity).

    \medskip
    \noindent
    \textbf{Limitation.}
    The main limitation of the above broadcast-based mitigation technique is that it will be ineffective if the adversary has strong control over the internal network.
    For instance, in Themis and Aequitas, it is assumed that the adversary controls all message delivery in the internal network, and can delay messages up to a bound $\Delta$.
    If $\Delta$ is large enough (e.g., if it is larger than the duration of the Condorcet attack) then the adversary can circumvent the proposed mitigation by delaying all the broadcast transactions so they are delivered only after the attack is complete.

\section{Simulation}
\label{sec:simulation}

    To assess the impact of the Condorcet attack, as well as the effectiveness of the proposed mitigation methods, we conduct a series of experiments through simulations. In this section, we present the results of these experiments.
        
    \textbf{Environments.}
    Our simulation encompasses four environments.
    The first environment captures the honest setting, where all the nodes and clients are honest, thereby eliminating the possibility of a Condorcet attack.
    Even in this environment, Condorcet cycles can occur. 
    Therefore, we are interested to know if our proposed ranked pairs batch-order scheme can more effectively order transactions within a cycle than the Hamiltonian-cycle-based scheme used in Themis.
    
    In the second environment, all the nodes in the system are honest, but there is an external adversary, who conducts the Condorcet attack from outside the system.
    In this environment, we are interested to evaluate the success rate and impact of the Condorcet attack (i.e., how many honest transactions the adversary can trap within a cycle).
    
    In the third environment, we introduce packet reordering to the external network.
    We evaluate the impact of this on the success rate of the Condorcet attack. 
    We also observe how the cloning method can help the adversary to improve its success rate.
    
    The last environment that we consider is similar to the second environment, except this time we guard the system using the proposed mitigation methods.
    In this environment, we measure the impact of the Condorcet attack in order to examine the strength of the proposed mitigation methods.

\textbf{Clients.}
  We use a sending process to submit all the clients' transactions to the system.
  The sending process transmits transactions in sequence at discrete times $t_i$, $i \geq 0$.
  At each time instance, the process sends ($n$ copies of) the transaction of a given client to all the $n$ nodes in the system.
  Each copy of the transaction will arrive at its destination node with a random delay drawn independently from a distribution named $\texttt{NetworkDist}$. 
  We refer to this distribution as the network latency. 
  We use another distribution, $\texttt{GenerationDist}$, to determine the delay between two consecutive time instances (i.e. $t_{i + 1} - t_i$ follows the $\texttt{GenerationDist}$ distribution).
  Similar to~\cite{Themis}, we set both $\texttt{GenerationDist}$ and $\texttt{NetworkDist}$ to exponential distributions with means of one and $r$, respectively.
  We refer to $r$ as \textit{external network ratio}.
  One can think of $r$ as the expected number of clients who transmit transactions within a time frame equal to the average network latency.

\subsection{Honest Environment}
    In this environment, all the nodes and clients are honest, and consequently, there is no Condorcet attack.
    Nevertheless, as shown in Figure~\ref{fig:honestEnv}, Condorcet cycles can occur particularly when the external network ratio is greater than one.

    To obtain the results plotted in Figure~\ref{fig:honestEnv}, we varied the external network ratio from 0.01 to 1000.
    For each given network ratio, and each network size of $n = 21$ and $n = 101$, we conducted 100 simulation runs.
    In each run, the sending process transmitted 100 transactions (at 100-time instances drawn from the $\texttt{GenerationDist}$ distribution).
    Once every node received all the transmitted transactions, we proceeded to generate the dependency graph using the Themis algorithm.
    By examining the graph (i.e. extracting strongly connected components) we then identified all the Condorcet cycles. 

    An interesting observation from Figure~\ref{fig:honestEnv} is that when the external network ratio is less than about one, Condorcet cycles rarely occur.
    As the external network ratio becomes larger than one, however, Condorcet cycles start to appear.
    For high values of the external network ratio, as depicted in Figure~\ref{fig:honestEnv}, Condorcet cycles not only occur frequently, but also include many of the transmitted transactions.
    Overall, this observation suggests a critical threshold at which the system's behavior, with respect to creating Condorcet cycles, significantly changes.


    \begin{figure*}
    \centering
    \begin{subfigure}[b]{.48\textwidth}
        \centering
         \includegraphics[width=\textwidth]{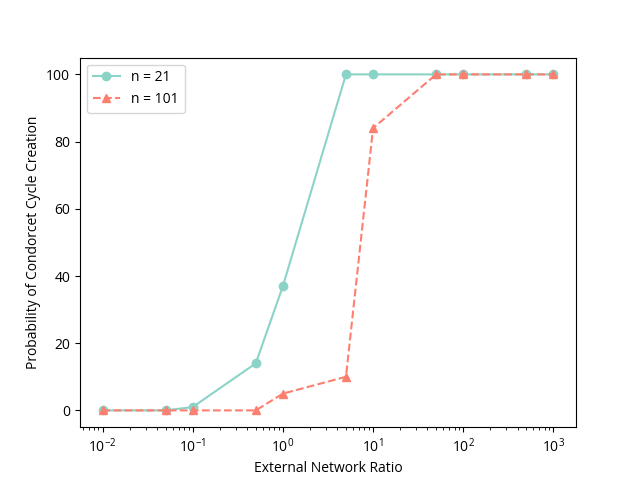}
         \subcaption{The chance of a Condorcet cycle}
         \label{fig:Cprob}
    \end{subfigure}
    \hfill
    \begin{subfigure}[b]{.48\textwidth}
        \centering
         \includegraphics[width=\textwidth]{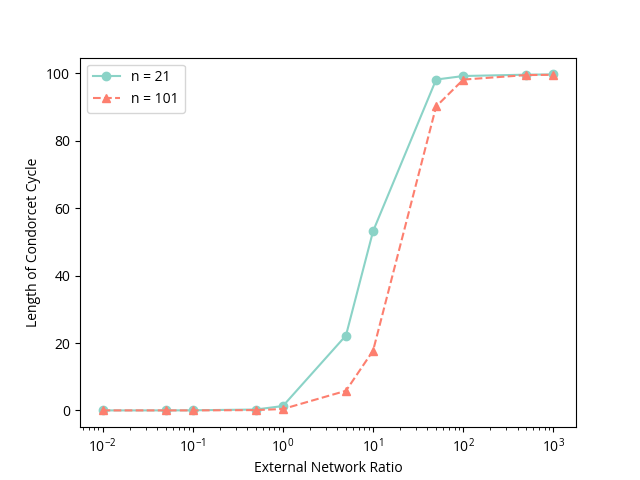}
         \subcaption{Number of transactions in cycles}
         \label{fig:Clength}
    \end{subfigure}
        \caption{Condorcet cycles in the honest environment
        }
        \label{fig:honestEnv}
    \end{figure*}

    We refer to Condorcet cycles that are not created by an adversary as \emph{natural Condorcet cycles}.
    Conversely, we call a Condorcet cycle adversarial if it is created by an adversary.
    In Section~\ref{sec:ranked}, we proposed a ranked pairs batch-ordering scheme to handle the ordering of transactions within an adversarial Condorcet cycle.
    Later in this section, we demonstrate that the proposed scheme indeed alleviates the severity of the Condorcet attack. 
    Here, we show (Figure~\ref{fig:fracHonest}) that the proposed ranked pairs batch-ordering scheme is also a good candidate for ordering transactions within a natural Condorcet cycle.
    Consequently, even in an honest environment, we can improve fairness in ordering transactions by replacing the existing batch-ordering schemes (i.e., the alphabetical scheme, and the Hamiltonian-based scheme of Themis) with the proposed ranked pairs batch-ordering scheme. 

    In Figure~\ref{fig:fracHonest}, the external network ratio (the $x$-axis) ranges from 1 to 1000; this is the range in which Condorcet cycles naturally occur.
    The $y$-axis shows the fraction of transaction pairs that are ordered correctly according to their transmission time.
    Each data point in Figure~\ref{fig:fracHonest} is the average of values obtained over 100 simulation runs.
    The data presented in this figure demonstrate the superiority of the proposed ranked pairs batch-ordering scheme for two network sizes of $n = 21$ and $n = 101$.


    \begin{figure*}[h]
    \centering
    \begin{subfigure}[b]{.48\textwidth}
        \centering
         \includegraphics[width=\textwidth]{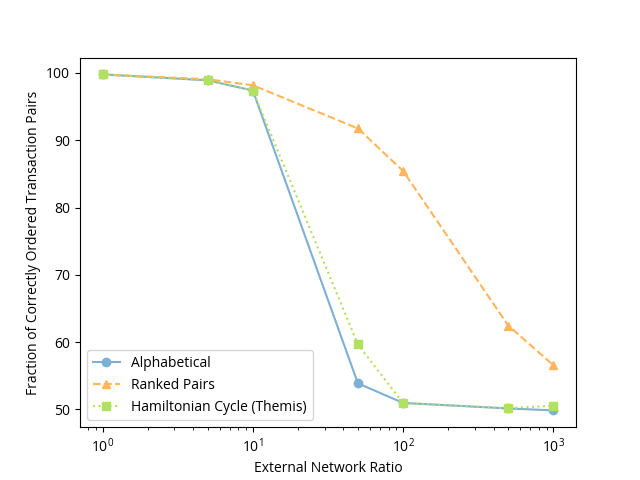}
         \subcaption{The number of nodes is $n = 21$}
         \label{fig:n21Honest}
    \end{subfigure}
    \hfill
    \begin{subfigure}[b]{.48\textwidth}
        \centering
         \includegraphics[width=\textwidth]{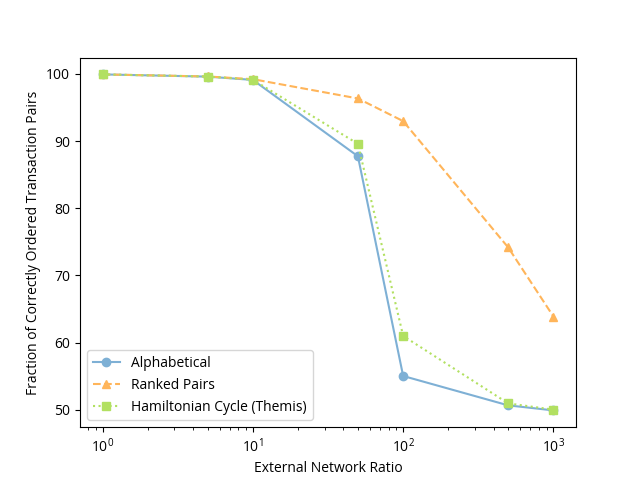}
         \subcaption{The number of nodes is $n = 101$}
         \label{fig:n101Honest}
    \end{subfigure}
        \caption{Fraction of correctly ordered transactions in the honest environment}
        \label{fig:fracHonest}
    \end{figure*}

\subsection{Adversarial Environment}
    In the existing adversarial environments in the literature, there is often at least one (typically up to $f = \theta(n)$) adversarial node in the system.
    In our adversarial environment, in contrast, all the nodes in the system can be honest.
    There is, however, an adversarial client in our environment who orchestrates the Condorcet attack from outside the system.
    
    In this section, we evaluate the performance of the Condorcet attack in this environment.
    In particular, we measure the success rate of the attack in the number of honest transactions it can trap within a cycle.
    The measurement is carried out for external network ratios $r$ less than one, as natural Condorcet cycles are rare in this regime, particularly when $r \ll 1$.
    This allows us to assess the strength of the attack in creating cycles in a setting where Condorcet cycles do not naturally happen.

    In our simulation, we simply use two adversarial transactions to create the Condorcet cycle as described in Example~\ref{exm:Cond3nodes}. 
    We set the pause time of the Condorcet attack to $\tau \in \{10, 50\}$ times the mean of the $\texttt{GenerationDist}$ distribution.
    This means that, on average, $\tau$ honest transactions are transmitted to the system during the pause time.    
    

    In parallel to the transmissions of honest transactions, the two adversarial transactions are transmitted to create a Condorcet cycle.
    Once all transactions are received by the nodes, we calculate two separate dependency graphs:
    one considering the adversarial transactions, and one ignoring them.
    By comparing these two dependency graphs, we then assess the impact of the attack on the final ordering.    

    Figures~\ref{fig:pefromrAdvTau10} and \ref{fig:pefromrAdvTau50} show the average number of the honest transactions that the attack can trap within cycles over two different settings: $\tau = 10$ and $\tau = 50$.
    As shown, for a wide range of external network ratios, the attack can trap nearly all the honest transactions that are transmitted during the pause time (about 9 honest transactions in the setting $\tau = 10$, and nearly 49 honest transactions in the setting $\tau = 50$). 
    This demonstrates the strength of the attack, considering that, on average $\tau$ honest transactions are submitted to the system during the pause time (and the attack traps nearly all of them).


    \begin{figure*}[h]
    \centering
    \begin{subfigure}[b]{.48\textwidth}
        \centering
         \includegraphics[width=\textwidth]{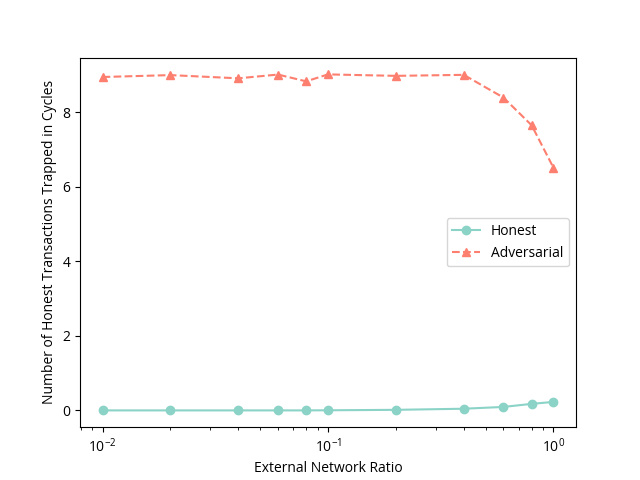}
         \subcaption{$\tau=10$, $n = 21$}
         \label{fig:n21tx20Adv}
    \end{subfigure}
    \hfill
    \begin{subfigure}[b]{.48\textwidth}
        \centering
         \includegraphics[width=\textwidth]{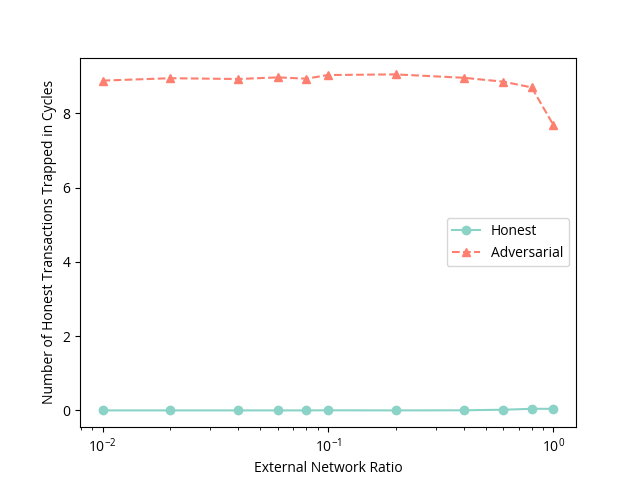}
         \subcaption{$\tau=10$, $n = 101$}
         \label{fig:n101tx20Adv}
    \end{subfigure}
        \caption{Number of honest transactions trapped in Condorcet cycles for $\tau=10$.}
        \label{fig:pefromrAdvTau10}
    \end{figure*}


    \begin{figure*}[h]
    \centering
    \begin{subfigure}[b]{.48\textwidth}
        \centering
         \includegraphics[width=\textwidth]{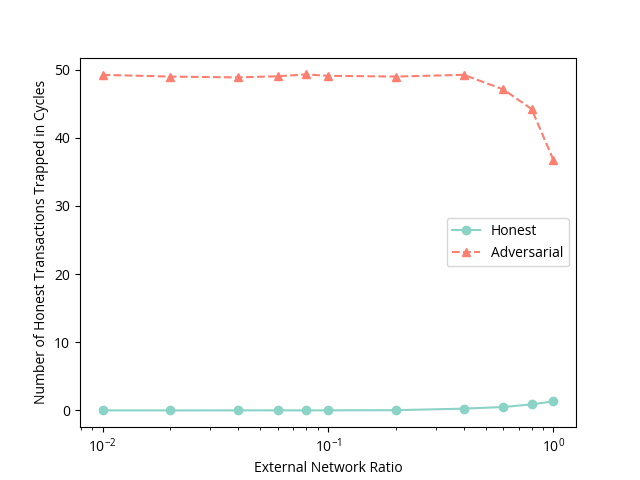}
         \subcaption{$\tau=50$, $n = 21$}
         \label{fig:t50n21}
    \end{subfigure}
    \hfill
    \begin{subfigure}[b]{.48\textwidth}
        \centering
         \includegraphics[width=\textwidth]{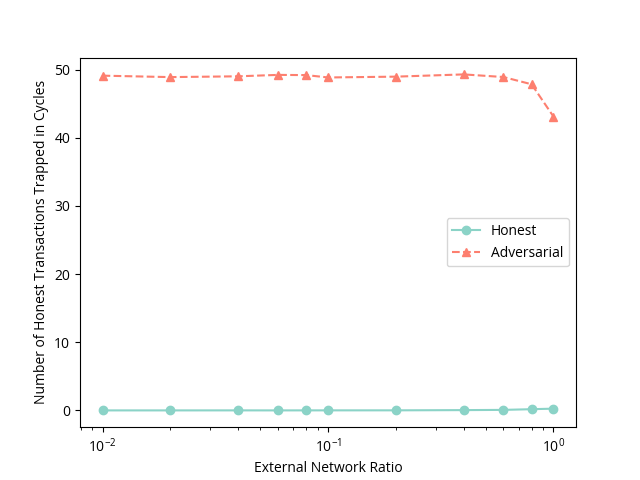}
         \subcaption{$\tau=50$, $n = 101$}
         \label{fig:nt50n101}
    \end{subfigure}
        \caption{Number of honest transactions trapped in Condorcet cycles for $\tau=50$.}
        \label{fig:pefromrAdvTau50}
    \end{figure*}
    
\subsection{Network Reordering}
\label{sec:cloning}
  In the Condorcet attack, the adversary sends a sequence of transactions in a particular order to create a cycle.
  The external network may, however, change the order of transactions transmitted, which can, in turn, reduce the attack's success rate.
  To evaluate this, we performed simulations over a network which changes the order of two consecutively transmitted transactions with probability $0 \leq p \leq 0.5$.
  For each value of $p$, we performed 1000 runs of simulations.
  The success rate of the attack was set to the fraction of runs in which the attack successfully trapped the honest transactions in a Condorcet cycle.
  
  Using the above setting, we conducted two instances of the Condorcet attack.
  The first instance uses two adversarial transactions $\texttt{A}$ and $\texttt{B}$ as in Example~\ref{exm:Cond3nodes}, and takes the following pattern:

  \[
    \begin{array}{cc}
      P_1: &  \texttt{A}, \texttt{B}, \text{Pause}\\
      P_2: &  \texttt{B}, \text{Pause}, \texttt{A}\\
      P_3: &  \text{Pause}, \texttt{A}, \texttt{B}\\      
    \end{array}
  \]

  \noindent
  As illustrated in Figure~\ref{fig:reorder}, this instance is sensitive to network reordering (the success rate of the attack drops quickly with $p$).
  As shown in the figure, the attack's success rate increases when we use the second instance of cloning described below.
  
  In our second instance (denote as $\texttt{tx} = 4$ in Figure~\ref{fig:reorder}), the adversary partitions nodes into four parts $P_1, P_2, P_3$ and $P_4$, and uses four transactions ($\texttt{A}$, $\texttt{B}$, $\texttt{C}$ and $\texttt{D}$) instead of two, in the following pattern:
  
  \[
    \begin{array}{cc}
      P_1: &  \texttt{A}, \texttt{B},\, \text{Pause} \,, \texttt{C}, \texttt{D}\\ 
      P_2: &  \texttt{B}, \texttt{C},\, \text{Pause} \,, \texttt{D}, \texttt{A}\\ 
      P_3: &  \texttt{C}, \texttt{D}, \, \text{Pause} \,, \texttt{A}, \texttt{B}\\  
      P_4: &  \texttt{D}, \texttt{A}, \, \text{Pause} \,, \texttt{B}, \texttt{C}\\              
    \end{array}
  \]  
  
  This instance of the Condorcet attack is more robust against network reordering as demonstrated in Figure~\ref{fig:reorder}.
  As in the first instance, the success rate of the instance can be boosted using the cloning method.
  In particular, note that the second instance together with a single clone is almost fully resistant to network transaction reordering.


    \begin{figure*}[h]
    \centering
    \begin{subfigure}[b]{.48\textwidth}
        \centering
         \includegraphics[width=\textwidth]{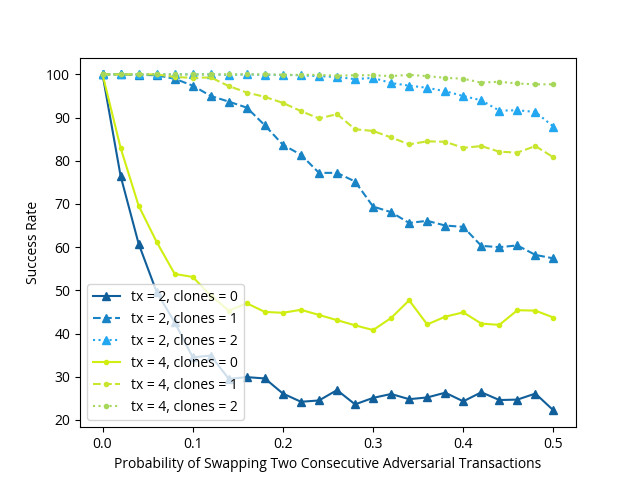}
         \subcaption{The number of nodes is $n = 21$}
         \label{fig:n21reorder}
    \end{subfigure}
    \hfill
    \begin{subfigure}[b]{.48\textwidth}
        \centering
         \includegraphics[width=\textwidth]{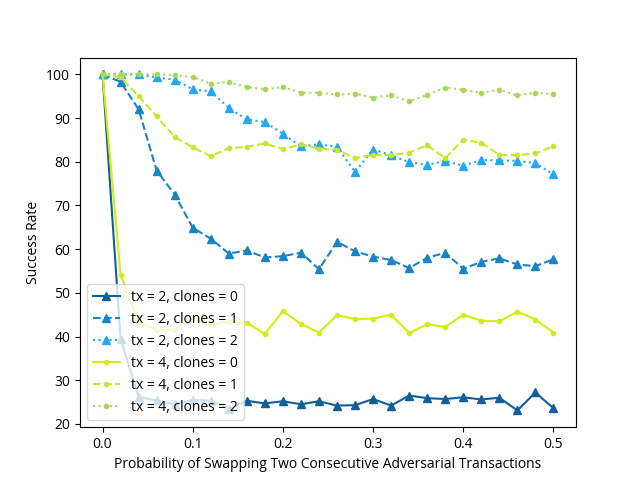}
         \subcaption{The number of nodes is $n = 101$}
         \label{fig:n101reorder}
    \end{subfigure}
        \caption{Impact of network reordering on the success of the Condorcet attack.}
        \label{fig:reorder}
    \end{figure*}

\subsection{A Non-Injective Condorcet Attack}
    Injecting transactions into the system is a key component of the proposed Condorcet attack.
    Without this component, an adversary has limited power in creating cycles even when the adversary controls the leader and a faction of all the nodes in the system.    

    To illustrate the above point, we conducted simulations over two networks with sizes: $n = 21$ and $n = 101$.
    In our simulation, the adversary controls the maximum fraction of nodes, including the leader, allowed by Themis (a quarter of nodes minus one). 
    All these nodes report the order of their received transactions in reverse, in a strategy to create Condorcet cycles\footnote{We note that this may not be an optimum strategy to create Condorcet cycles. Nevertheless, we believe that an optimum strategy (which may be computationally intractable) may not be significantly more successful than the adopted strategy. We leave the validation of this claim for future work.}.
    The external network ratio is varied from 0.01 to 100 to capture a wide range of network conditions.
    The total number of transmitted transactions is set to 100.     

    To evaluate the impact of the above strategy in creating cycles, we created two dependency graphs.
    The first graph represents the scenario where the adversarial nodes reverse their orderings, whereas the second graph represents the scenario where the adversarial nodes report the true ordering.
    Figure~\ref{fig:Internal} shows the results of our simulation.

    As shown in Figure~\ref{fig:Internal}, the adversary's attempts to create cycles are largely unsuccessful in the region where the external network ratio is less than one.
    We note that in this region, the average temporal gap between two different transaction transmissions is more than the average network latency.
    In particular, when $r \ll 1$ (i.e., when transactions are transmitted far apart in time with respect to the network latency), honest nodes in the system have a clear view of the true ordering of transactions.
    In this region, the adversary is all but powerless in creating cycles\footnote{When $r > 1$ (i.e., in the region where Condorcet cycles naturally emerge) the adversary achieves some degree of success in creating larger cycles than naturally occur.}, as evident in Figure~\ref{fig:Internal}.
    In contrast, in the same region, an external adversary can create a cycle using the proposed Condorcet attack, even when all the nodes in the system are honest.


    \begin{figure*}[h]
    \centering
    \begin{subfigure}[b]{.48\textwidth}
        \centering
         \includegraphics[width=\textwidth]{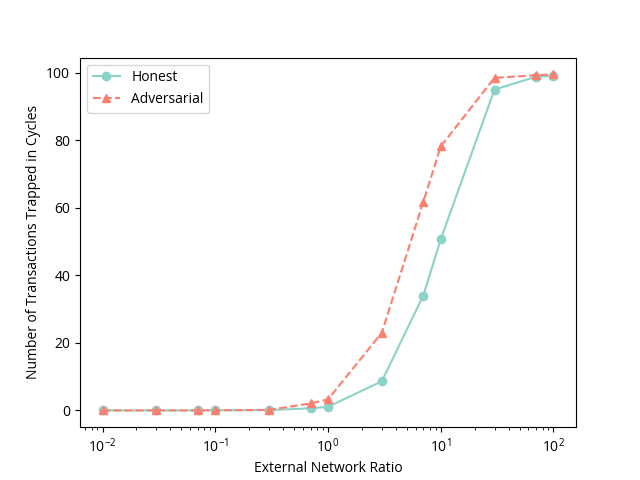}
         \subcaption{The number of nodes is $n=21$}
         \label{fig:n21Internal}
    \end{subfigure}
    \hfill
    \begin{minipage}[b]{.48\textwidth}
        \centering
         \includegraphics[width=\textwidth]{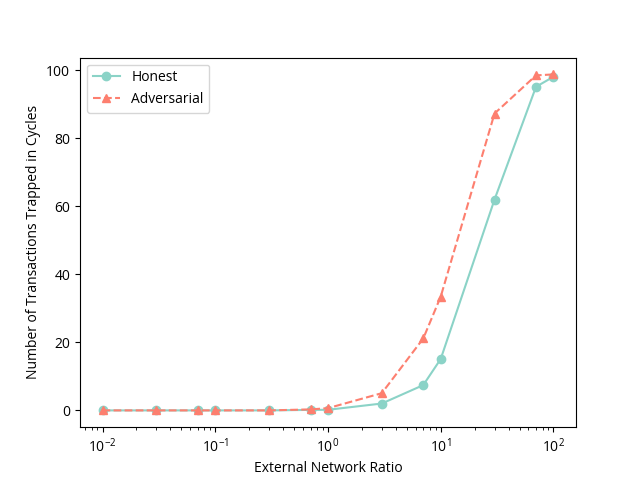}
         \subcaption{The number of nodes is $n=101$}
         \label{fig:n101Internal}
    \end{minipage}
        \caption{The non-injective attack has limited power in creating cycles.}     
        \label{fig:Internal}
    \end{figure*}

\subsection{Mitigation}
    \label{sec:mitigationsim}
    In this section, we evaluate the performance of our mitigation methods in preventing or minimizing the impact of the Condorcet attack. 

    \textbf{Ranked-pairs-based Mitigation Method.}
    To evaluate the effectiveness of this mitigation, we conducted a simulation over two network sizes of $n = 21$ and $n = 101$.
    We set the pause time of the attack to 10 times the mean of $\texttt{GenerationDist}$, and set the total number of honest transactions to~20.
    We varied the external network ratio $r$ from 0.001 to 1.
    Recall that in this range of external network ratio (i.e., $r < 1$), Condorcet cycles do not emerge naturally; rather they are created by the Condorcet attack.
    To evaluate the true impact of our ranked-pars mitigation method, therefore, we focused on this region.    
    
    Figure~\ref{fig:Ranked} compares the performance of our proposed ranked-pairs-based mitigation method to the Hamiltonian-based method used in Themis, and the simple alphabetical method.
    The results show that the proposed ranked-pairs method achieves a low error rate, indicating that it can effectively order honest transactions correctly even when they fall in a Condorcet cycle.
    In contrast, the Themis algorithm exhibits an error rate of approximately $31\%$.
    The error rate in the case of alphabetical ordering is $50\%$.
    Note that a random ordering method can, on average, correctly orders $50\%$ of all the pairs of transactions.
    In this sense, the worst-case transaction ordering error is $50\%$, which is the case for the alphabetical method (this method is essentially a random ordering method).


    \begin{figure*}[h]
    \centering
    \begin{subfigure}[b]{.48\textwidth}
        \centering
         \includegraphics[width=\textwidth]{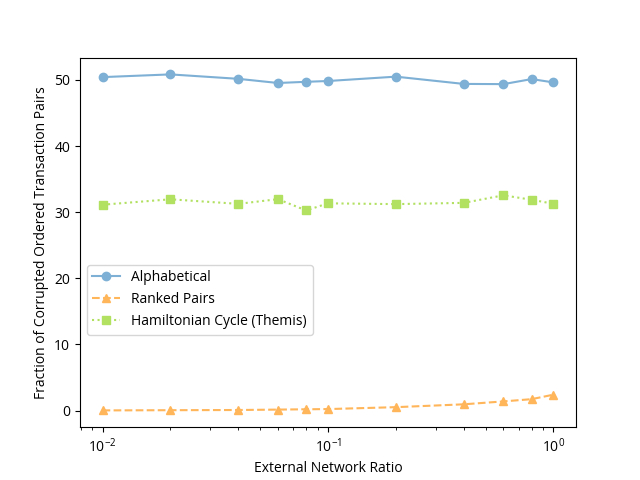}
         \subcaption{The number of nodes is $n=21$}
         \label{fig:n21Ranked}
    \end{subfigure}
    \hfill
    \begin{subfigure}[b]{.48\textwidth}
        \centering
         \includegraphics[width=\textwidth]{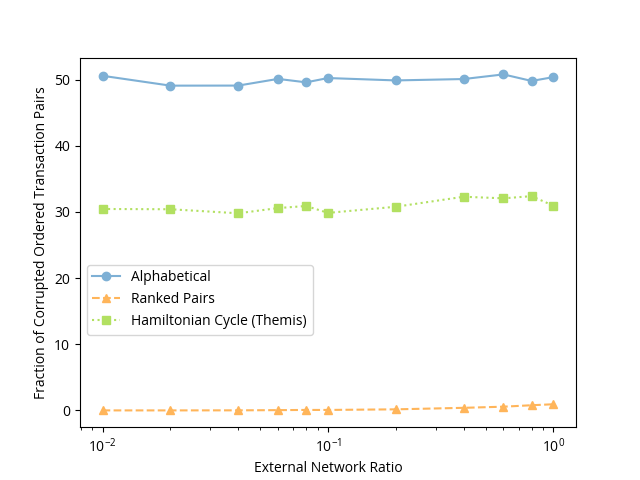}
         \subcaption{The number of nodes is $n=101$}
         \label{fig:n101Ranked}
    \end{subfigure}
        \caption{The performance of the proposed ranked pairs based mitigation method}
        \label{fig:Ranked}
    \end{figure*}

    \textbf{The Broadcast-based Mitigation Method.}
    To evaluate the effectiveness of the broadcast-based mitigation method, we conducted simulations using two network sizes: $n = 21$ and $n = 101$.
    We introduced a new exponential distribution called $\texttt{InternalNetworkDist}$, which captures the random delays experienced by messages within the internal network.
    Specifically, we sample from $\texttt{InternalNetworkDist}$ to determine the delay between sending a transaction from one node to another node.
    This is in contrast to $\texttt{NetworkDist}$, which is used to determine the random delays between a client and a node in the external network.

    In our simulation, we set the mean of $\texttt{InternalNetworkDist}$ to $r'$.
    We refer to $r'$ as the \emph{internal network ratio}. 
    In our simulations, we set $\tau$ to 10 times the mean of $\texttt{GenerationDist}$ (i.e. $\tau = 10 \cdot r)$, and set the total number of honest transactions to 20. 
    We fixed the external network ratio to $r = 0.1$, to ensure that no natural Condorcet cycles were created, and varied the internal network ratio $r'$ from 0.01 to 1000.
    
    We analyzed the number of honest transactions trapped in a Condorcet cycle under three different settings.
    In the first setting, referred to as the ``honest setting'', nodes did not broadcast and the adversary did not conduct a Condorcet attack.
    In the second setting, nodes still did not broadcast, but the adversary attempted a Condorcet attack.
    Finally, in the last setting, the adversary launched an attack while the nodes employed the broadcasting method to mitigate it.

    Figure~\ref{fig:Broadcast} shows the result of our simulations in the above three settings. 
    The results demonstrate that the proposed broadcast-based mitigation is highly effective in preventing the adversary from creating a Condorcet cycle and trapping honest transactions.
    This can be attributed to two key factors:    
    Firstly, the mitigation strategy disrupts the completion of the pause phase, thereby preventing honest transactions from being trapped in a Condorcet cycle. 
    When the internal network ratio $r'$ is smaller than the pause time, almost no transactions are trapped.
    Interestingly, even when $r'$ exceeds the pause time, the adversary cannot achieve the same level of performance. 
    It is because the broadcast of transactions with the internal network can still somewhat disturb the ordering of adversarial transactions.
    This reduces the success rate of the attack as the specific ordering of adversarial transactions is crucial for creating a Condorcet cycle. 
    If, on the other hand, the adversary has enough control over the internal network to delay transactions as much as the pause time, it can circumvent the proposed broadcast-based mitigation as the adversary can enforce the ordering of its transactions within the internal network by delaying all the messages.
    

    \begin{figure*}[h]
    \centering
    \begin{subfigure}[b]{.48\textwidth}
        \centering
         \includegraphics[width=\textwidth]{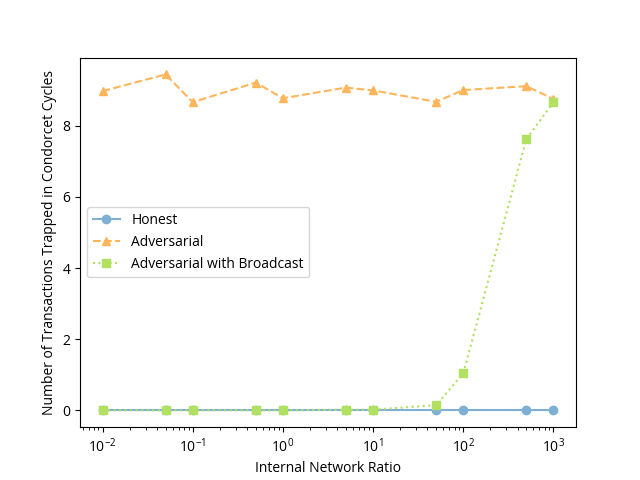}
         \subcaption{The number of nodes is $n=21$}
         \label{fig:n21Broadcast}
    \end{subfigure}
    \hfill
    \begin{subfigure}[b]{.48\textwidth}
        \centering
         \includegraphics[width=\textwidth]{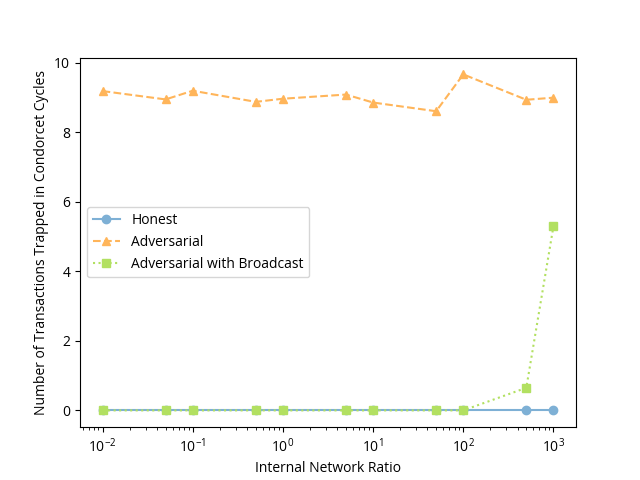}
         \subcaption{The number of nodes is $n=101$}
         \label{fig:n101Broadcast}
    \end{subfigure}
        \caption{The performance of the proposed broadcast mitigation method}
        \label{fig:Broadcast}
    \end{figure*}

\pagebreak
\section{Conclusion}
    Condorcet cycles can occur naturally. 
    While these natural cycles may not significantly disrupt fairness in the system since transactions falling within these cycles are typically received around the same time, the artificial creation of Condorcet cycles can lead to significant unfairness in the system. 
    In this paper, we showed that even with all nodes in the system behaving honestly, it is relatively simple to generate such artificial cycles. 
    Furthermore, we demonstrated that these created cycles possess significant power, as they can trap transactions submitted at widely different times that would not naturally fall within a cycle.
    
    To address this attack, we proposed three mitigation methods using different approaches. 
    These methods complement one another and can be employed collectively to fortify the defensive measures against the attack. 
    Through simulations, we showcased that despite their described limitations, the proposed mitigation methods can substantially reduce the adverse impact of the Condorcet attack.



\noindent
\bibliography{sample}

\begin{thebibliography}{10}

\bibitem{BaumCDFG21}
Carsten Baum, James~Hsin{-}yu Chiang, Bernardo David, Tore~Kasper Frederiksen,
  and Lorenzo Gentile.
\newblock Sok: Mitigation of front-running in decentralized finance.
\newblock {\em {IACR} Cryptol. ePrint Arch.}, page 1628, 2021.
\newblock URL: \url{https://eprint.iacr.org/2021/1628}.

\bibitem{SocialChoice}
Felix Brandt, Vincent Conitzer, Ulle Endriss, J{\'{e}}r{\^{o}}me Lang, and
  Ariel~D. Procaccia, editors.
\newblock {\em Handbook of Computational Social Choice}.
\newblock Cambridge University Press, 2016.
\newblock \href {https://doi.org/10.1017/CBO9781107446984}
  {\path{doi:10.1017/CBO9781107446984}}.

\bibitem{CachinKPS01}
Christian Cachin, Klaus Kursawe, Frank Petzold, and Victor Shoup.
\newblock Secure and efficient asynchronous broadcast protocols.
\newblock In Joe Kilian, editor, {\em Advances in Cryptology - {CRYPTO} 2001,
  21st Annual International Cryptology Conference, Santa Barbara, California,
  USA, August 19-23, 2001, Proceedings}, volume 2139 of {\em Lecture Notes in
  Computer Science}, pages 524--541. Springer, 2001.
\newblock \href {https://doi.org/10.1007/3-540-44647-8\_31}
  {\path{doi:10.1007/3-540-44647-8\_31}}.

\bibitem{QuickFairness}
Christian Cachin, Jovana Micic, Nathalie Steinhauer, and Luca Zanolini.
\newblock Quick order fairness.
\newblock In {\em Financial Cryptography and Data Security - 26th International
  Conference, {FC} 2022, Grenada, May 2-6, 2022, Revised Selected Papers},
  Lecture Notes in Computer Science, pages 316--333. Springer, 2022.
\newblock \href {https://doi.org/10.1007/978-3-031-18283-9\_15}
  {\path{doi:10.1007/978-3-031-18283-9\_15}}.

\bibitem{Condorcet}
M.~d.~Condorcet.
\newblock Essay on the application of analysis to the probability of majority
  decisions.
\newblock {\em Paris: Imprimerie Royale}, 1785.

\bibitem{FlashBoys}
Philip Daian, Steven Goldfeder, Tyler Kell, Yunqi Li, Xueyuan Zhao, Iddo
  Bentov, Lorenz Breidenbach, and Ari Juels.
\newblock Flash boys 2.0: Frontrunning, transaction reordering, and consensus
  instability in decentralized exchanges.
\newblock {\em CoRR}, abs/1904.05234, 2019.
\newblock URL: \url{http://arxiv.org/abs/1904.05234}, \href
  {http://arxiv.org/abs/1904.05234} {\path{arXiv:1904.05234}}.

\bibitem{PartialSynchrony}
Cynthia Dwork, Nancy~A. Lynch, and Larry~J. Stockmeyer.
\newblock Consensus in the presence of partial synchrony.
\newblock {\em J. {ACM}}, 35(2):288--323, 1988.
\newblock \href {https://doi.org/10.1145/42282.42283}
  {\path{doi:10.1145/42282.42283}}.

\bibitem{EskandariMC19}
Shayan Eskandari, Seyedehmahsa Moosavi, and Jeremy Clark.
\newblock Sok: Transparent dishonesty: Front-running attacks on blockchain.
\newblock In Andrea Bracciali, Jeremy Clark, Federico Pintore, Peter~B.
  R{\o}nne, and Massimiliano Sala, editors, {\em Financial Cryptography and
  Data Security - {FC} 2019 International Workshops, {VOTING} and WTSC, St.
  Kitts, St. Kitts and Nevis, February 18-22, 2019, Revised Selected Papers},
  volume 11599 of {\em Lecture Notes in Computer Science}, pages 170--189.
  Springer, 2019.
\newblock \href {https://doi.org/10.1007/978-3-030-43725-1\_13}
  {\path{doi:10.1007/978-3-030-43725-1\_13}}.

\bibitem{abs-2203-11520}
Lioba Heimbach and Roger Wattenhofer.
\newblock Sok: Preventing transaction reordering manipulations in decentralized
  finance.
\newblock {\em CoRR}, abs/2203.11520, 2022.
\newblock \href {http://arxiv.org/abs/2203.11520} {\path{arXiv:2203.11520}},
  \href {https://doi.org/10.48550/arXiv.2203.11520}
  {\path{doi:10.48550/arXiv.2203.11520}}.

\bibitem{asiapkc/KelkarDK22}
Mahimna Kelkar, Soubhik Deb, and Sreeram Kannan.
\newblock Order-fair consensus in the permissionless setting.
\newblock In Jason~Paul Cruz and Naoto Yanai, editors, {\em {APKC} '22:
  Proceedings of the 9th {ACM} on {ASIA} Public-Key Cryptography Workshop,
  APKC@AsiaCCS 2022, Nagasaki, Japan, 30 May 2022}, pages 3--14. {ACM}, 2022.
\newblock \href {https://doi.org/10.1145/3494105.3526239}
  {\path{doi:10.1145/3494105.3526239}}.

\bibitem{Themis}
Mahimna Kelkar, Soubhik Deb, Sishan Long, Ari Juels, and Sreeram Kannan.
\newblock Themis: Fast, strong order-fairness in byzantine consensus.
\newblock {\em {IACR} Cryptol. ePrint Arch.}, page 1465, 2021.
\newblock URL: \url{https://eprint.iacr.org/2021/1465}.

\bibitem{Aequitas}
Mahimna Kelkar, Fan Zhang, Steven Goldfeder, and Ari Juels.
\newblock Order-fairness for byzantine consensus.
\newblock In Daniele Micciancio and Thomas Ristenpart, editors, {\em Advances
  in Cryptology - {CRYPTO} 2020 - 40th Annual International Cryptology
  Conference, {CRYPTO} 2020, Santa Barbara, CA, USA, August 17-21, 2020,
  Proceedings, Part {III}}, volume 12172 of {\em Lecture Notes in Computer
  Science}, pages 451--480. Springer, 2020.
\newblock \href {https://doi.org/10.1007/978-3-030-56877-1\_16}
  {\path{doi:10.1007/978-3-030-56877-1\_16}}.

\bibitem{Wendy}
Klaus Kursawe.
\newblock Wendy, the good little fairness widget: Achieving order fairness for
  blockchains.
\newblock In {\em {AFT} '20: 2nd {ACM} Conference on Advances in Financial
  Technologies, New York, NY, USA, October 21-23, 2020}, pages 25--36. {ACM},
  2020.
\newblock \href {https://doi.org/10.1145/3419614.3423263}
  {\path{doi:10.1145/3419614.3423263}}.

\bibitem{Harvard}
David~C. Parkes and Lirong Xia.
\newblock A complexity-of-strategic-behavior comparison between schulze's rule
  and ranked pairs.
\newblock In J{\"{o}}rg Hoffmann and Bart Selman, editors, {\em Proceedings of
  the Twenty-Sixth {AAAI} Conference on Artificial Intelligence, July 22-26,
  2012, Toronto, Ontario, Canada}. {AAAI} Press, 2012.
\newblock URL:
  \url{http://www.aaai.org/ocs/index.php/AAAI/AAAI12/paper/view/5075}.

\bibitem{DarkForest}
Kaihua Qin, Liyi Zhou, and Arthur Gervais.
\newblock Quantifying blockchain extractable value: How dark is the forest?
\newblock In {\em 43rd {IEEE} Symposium on Security and Privacy, {SP} 2022, San
  Francisco, CA, USA, May 22-26, 2022}, pages 198--214. {IEEE}, 2022.
\newblock \href {https://doi.org/10.1109/SP46214.2022.9833734}
  {\path{doi:10.1109/SP46214.2022.9833734}}.

\bibitem{ReiterBirman}
Michael~K. Reiter and Kenneth~P. Birman.
\newblock How to securely replicate services.
\newblock {\em ACM Transactions on Programming Languages and Systems},
  16(3):986–1009, 1994.
\newblock \href {https://doi.org/10.1145/177492.177745}
  {\path{doi:10.1145/177492.177745}}.

\bibitem{Schulze11}
Markus Schulze.
\newblock A new monotonic, clone-independent, reversal symmetric, and
  condorcet-consistent single-winner election method.
\newblock {\em Soc. Choice Welf.}, 36(2):267--303, 2011.
\newblock \href {https://doi.org/10.1007/s00355-010-0475-4}
  {\path{doi:10.1007/s00355-010-0475-4}}.

\bibitem{RankedPairs}
T.~N. Tideman.
\newblock Independence of clones as a criterion for voting rules.
\newblock {\em Social Choice and Welfare}, 4(3):185--206, September 1987.
\newblock \href {https://doi.org/10.1007/bf00433944}
  {\path{doi:10.1007/bf00433944}}.

\bibitem{Pompe}
Yunhao Zhang, Srinath T.~V. Setty, Qi~Chen, Lidong Zhou, and Lorenzo Alvisi.
\newblock Byzantine ordered consensus without byzantine oligarchy.
\newblock In {\em 14th {USENIX} Symposium on Operating Systems Design and
  Implementation, {OSDI} 2020, Virtual Event, November 4-6, 2020}, pages
  633--649. {USENIX} Association, 2020.
\newblock URL:
  \url{https://www.usenix.org/conference/osdi20/presentation/zhang-yunhao}.

\end{thebibliography}
\end{document}